\documentclass[12pt]{article}
\usepackage{amsmath,amsthm,amssymb}
\usepackage{array, longtable}
\usepackage{pdfsync}
\newtheorem{Theorem}{Theorem}[section]
\newtheorem{lem}[Theorem]{Lemma}
\newtheorem{Remark}[Theorem]{Remark}
\newtheorem{Definition}[Theorem]{Definition}
\newtheorem{Corollary}[Theorem]{Corollary}
\newtheorem{Proposition}[Theorem]{Proposition}
\newtheorem{Example}[Theorem]{Example}

\numberwithin{equation}{section}
\numberwithin{table}{section}

\begin{document}

\title{  New bounds for $b$-Symbol Distances of Matrix Product Codes
\footnote{E-Mail addresses: panxucode@163.com (X. Pan), lingsan@ntu.edu.sg (S. Ling),
hwliu@mail.ccnu.edu.cn (H. Liu).}}

\author{Xu Pan$^1$,~San Ling$^2$,~Hongwei Liu$^3$}
\date{\small
$^1$School of Information Science and Technology/Cyber Security,
Jinan University, Guangzhou, Guangdong, 510632, China\\
$^2$School of Physical and Mathematical Sciences, Nanyang Technological University, 637371, Singapore\\
$^3$School of Mathematics and Statistics, Central China Normal University, Wuhan, Hubei, 430079, China\\
}
\maketitle

\begin{abstract}
Matrix product codes are generalizations of some well-known constructions of codes, such as Reed-Muller codes, $[u+v,u-v]$-construction, etc.  Recently, a bound for the symbol-pair distance of a matrix product code was given in \cite{LEL}, and new families of MDS symbol-pair codes were constructed by using this bound. In this paper, we generalize this bound to the $b$-symbol distance of a matrix product code and determine all minimum $b$-symbol distances of Reed-Muller codes. We also give a bound for the minimum $b$-symbol distance of codes obtained from the $[u+v,u-v]$-construction, and use this bound to construct some $[2n,2n-2]_q$-linear $b$-symbol almost MDS codes with arbitrary length. All the minimum $b$-symbol distances of $[n,n-1]_q$-linear codes and $[n,n-2]_q$-linear codes for $1\leq b\leq n$ are determined. Some examples are presented to illustrate these results.


\end{abstract}

\section{Introduction}

In 2011, Cassuto and Blaum \cite{CB} introduced a new metric framework, named the symbol-pair distance, to protect against pair errors in symbol-pair read channels, where the outputs are overlapping pairs of symbols. In \cite{CJK}, Chee {\it et al.} established a Singleton-like bound for symbol-pair codes and constructed some MDS symbol-pair codes which are linear codes meeting this Singleton-like bound.
Several MDS symbol-pair codes have been constructed (see, for example,  \cite{CJK}, \cite{CLL}, \cite{D}, \cite{KZL}, \cite{KZZLC}, \cite{LG}, \cite{LEL}, \cite{ML1} and \cite{ML2}).
In \cite{CL}, some new bounds on the code size of symbol-pair codes (not necessarily linear) over $\mathbb{F}_q$,
where $\mathbb{F}_q$ is the finite field of order $q$,
were proved by using the theory of classical codes over $\mathbb{F}_{q^2}$.
 In \cite{DNSS} and \cite{DWLS},  the symbol-pair distances of repeated-root constacyclic codes of lengths $p^s$ and $2p^s$ were calculated.
Generalized pair weights of linear codes are generalizations of minimum symbol-pair weights, which were introduced by Liu and Pan \cite{LP}. Generalized pair weights also can be used to characterize the ability of a code for protecting information in the symbol-pair read wire-tap channels of type II.

In 2016, Yaakobi, Bruck and Siegel \cite{YBS} generalized the symbol-pair weight to the $b$-symbol weight.
The Singleton-like bound $$d_{b}(C)\leq \min\{n-k+b,n\}$$ for the minimum $b$-symbol weight $d_{b}(C)$ of an $[n,k]_q$-linear code $C$ was proved in \cite{DZ} and \cite{LP2}.
An $[n,k]_q$-linear code $C$ satisfying $$d_{b}(C)= \min\{n-k+b,n\}$$ is called a $b$-symbol MDS code.
In \cite{LP2}, Liu and Pan provided a necessary and sufficient condition for a linear code to be a $b$-symbol MDS code by using a generator matrix or a parity check matrix of this linear code. Liu and Pan also introduced the notion of the generalized $b$-symbol weight of a linear code, which is a generalization of the generalized Hamming weight and the generalized pair weight, and they obtained some basic properties and bounds of the generalized $b$-symbol weight in \cite{LP2}.
In \cite{SO}, Shi {\it et al}. used algebraic curves over finite fields to obtain tight lower and upper bounds on the $b$-symbol weight of arbitrary cyclic codes.
The $b$-symbol weight distribution of some irreducible cyclic codes was determined in \cite{V} and \cite{ZSO}.

In \cite{BN}, the matrix product code $[C_{1},\cdots, C_{M}]\cdot A$ was introduced, where $C_{1},\cdots,C_{M}$ are codes of length $n$ over $\mathbb{F}_{q}$ and $A=(a_{i,j})_{M\times N}$ is an $M\times N$ matrix over $\mathbb{F}_{q}$.
Matrix product codes are generalizations of the Reed-Muller codes and codes obtained by the  $[u+v,u-v]$-construction, $[u+v+w,2u+v,u]$-construction, $[a + x,b + x,a + b + x]$-construction and $[u+v,u-v]$-construction, and etc.
The linear codes obtained by the $[u+v,u-v]$-construction have many good properties, for example, the linear code
$$C=[C_1,C_2]\cdot\left(\begin{array}{cc}
1 &1\\

1& -1
\end{array}\right) =\{[u+v,u-v]\,|\,u\in C_1,\,v\in C_2\}$$
is a cyclic code when $C_1$ is a cyclic code and $C_2$ is a negacyclic code by Theorem 8.1 of \cite{H}.

In \cite{CCDFM}, Cao {\it et al}.  proved that any $\lambda^{p^{k}}$-constacyclic code of length $p^kn$ over  $\mathbb{F}_{p^m}$ is monomially equivalent to a matrix product code of a nested sequence of $p^k\lambda$-constacyclic codes of length $n$ over $\mathbb{F}_{p^m}$, where $\lambda$ is a nonzero element of $\mathbb{F}_{p^m}$.

 A lower bound for the minimum Hamming distance of matrix product codes over finite fields was obtained in \cite{OS}.
Decoding methods for some matrix product codes were also discussed in \cite{HR} and \cite{HL}. Results about matrix product codes over finite commutative rings were obtained in \cite{FLL}, \cite{FLL1}, \cite{LL} and \cite{V}.

Generalized Reed-Muller codes were studied by \cite{DG}, \cite{KLP}, \cite{KLP1} and \cite{KLP2}.
In \cite{BC}, the automorphism group of generalized Reed-Muller codes is given.
Romanov studied relations between single error correcting perfect codes and Reed-Muller codes, and showed that affine Reed-Muller codes of order $(q-1)m-2$ are quasi-perfect codes in \cite{R}.
A description of generalized Hamming weights of generalized Reed-Muller codes was provided in \cite{BC1} and \cite{HP}.

The main goal of this paper is to study the minimum $b$-symbol weight of matrix product codes. Our main contributions are as follows:
\begin{itemize}
  \item We give some tight lower bounds for the minimum $b$-symbol weight of any matrix product code $[C_{1},\cdots, C_{M}]\cdot A$ for $1\leq b\leq n$ in \textbf{Theorem~\ref{t3.1}}. When $A=(a_{i,j})_{M\times N}$ is an upper triangular nonsingular by column matrix, an upper bound for the minimum $b$-symbol weight of the matrix product code $[C_{1},\cdots, C_{M}]\cdot A$ is obtained for $1\leq b\leq n$  in \textbf{Corollary~\ref{t3.2}}.

  \item We determine all the minimum $b$-symbol weights of Reed-Muller codes in \textbf{Theorem~\ref{weight}}. We provide a necessary and sufficient condition for a Reed-Muller code to be a $b$-symbol MDS code in \textbf{Corollary~\ref{4.7}}.
   \item We study the minimum $b$-symbol weight of linear codes obtained by the $[u+v,u-v]$-construction in \textbf{Theorem~\ref{5.1}}. It is used to construct some $[2n, 2n-2]_q$-linear $b$-symbol almost MDS codes, where $n$ can be any large integer, in \textbf{Corollary~\ref{1.9a}} and \textbf{Remark~\ref{5.6}}.
  \item We determine all the minimum $b$-symbol weights of $[n, n-1]_q$-linear codes and $[n, n-2]_q$-linear codes in \textbf{Theorem~\ref{1.10}} and \textbf{Theorem~\ref{1.11}} respectively, and we give some examples to illustrate these results.
\end{itemize}

The rest of the paper is organized as follows: In Section 2, we give some preliminaries and some notations. In Section 3, tight lower and upper bounds for the minimum $b$-symbol weight of matrix product codes are obtained.  In Section 4, we determine all the minimum $b$-symbol weights of Reed-Muller codes. In Section 5, we give lower and upper bounds for the minimum $b$-symbol weight of linear codes obtained by the $[u+v,u-v]$-construction. As an application of the $[u+v,u-v]$-construction, we construct some $[2n, 2n-2]_q$-linear $b$-symbol AMDS codes where $n$ can be any large integer.
In Section 6, we obtain all the minimum $b$-symbol weights of $[n, n-1]_q$-linear codes and $[n, n-2]_q$-linear codes.

\section{Preliminaries}

Throughout this paper, let $\mathbb{F}_q$ be the finite field of order $q$, where $q=p^e$ and $p$ is a prime.
For $n\in \mathbb{N}^+=\{ 1,2,\cdots\}$, let $\mathbb{F}_q^{n}$ be the $n$-dimensional vector space over $\mathbb{F}_q$.
A nonempty subset of $\mathbb{F}_q^{n}$ is called a  {\it code} of length $n$ over $\mathbb{F}_{q}$.
An $\mathbb{F}_q$-subspace $C$ of dimension $k$ of $\mathbb{F}_q^{n}$ is called an {\it $[n,k]_q$-linear code}. The dual code $C^{\perp}$ of $C$ is defined as
$$
C^{\perp}=\{{\bf x}\in \mathbb{F}_q^n \,| \,  {\bf c}\cdot{\bf x} =0, \forall \, {\bf c}\in C\},
$$
where $ ``-\cdot-" $ denotes the  {\it Euclidean inner product}.

For $n,b\in \mathbb{N}^+$, we always assume $ 1\leq b \leq n$ and $\mathbb{Z}_n=\{1,2,\cdots,n\}$ with zero element
$$n\equiv0\mod n.$$

\begin{Definition}\label{b weight}
Let $\mathbf{x}=(x_{1},x_{2},\cdots,x_{n})\in  \mathbb{F}_q^{n}$. The $b$-symbol support of $\mathbf{x}$, denoted by $\chi_{b}(\mathbf{x})$, is $$\chi_{b}(\mathbf{x})=\{ i \in\mathbb{Z}_{n} \,|\, (x_{i},x_{i+1},\cdots,x_{i+b-1})\neq(0,0,\cdots,0)\}.$$
 The $b$-symbol weight of $\mathbf{x}$ is defined as $w_{b}(\mathbf{x})=|\chi_{b}(\mathbf{x})|$.
\end{Definition}

Since we assume $\mathbb{Z}_n=\{1,2,\cdots,n\}$, we can view $\chi_{b}(\mathbf{x})$ as a subset of $\mathbb{Z}_n$ for any $\mathbf{x}\in  \mathbb{F}_q^{n}$.
When $b=1$, we know that $\chi_{1}(\mathbf{x})$ is the Hamming support of $\mathbf{x}$ and $w_{1}(\mathbf{x})$ is the Hamming weight of $\mathbf{x}$.

\begin{Definition}\label{AA}(\cite{YBS})
For any $\mathbf{x},\mathbf{y} \in \mathbb{F}_{q}^{n}$, the  $b$-symbol distance between $\mathbf{x}$ and $\mathbf{y}$ is defined as
$$d_{b}(\mathbf{x},\mathbf{y})=w_{b}(\mathbf{x}-\mathbf{y}).$$
\end{Definition}

Let $C$ be an $[n,k]_q$-linear code, the {\it minimum $b$-symbol distance} of $C$ is defined as
 $$d_{b}(C)=\min_{\mathbf{c}\ne\mathbf{c'} \in C}\,d_{b}(\mathbf{c},\mathbf{c}')=\min_{{\bf 0}\ne \mathbf{c}\in C}\,w_{b}(\mathbf{c}).$$
In Theorem 3.4 of \cite{LP2} and in \cite{DZ}, the authors gave the Singleton-like bound $$d_{b}(C)\leq\min\{n-k+b,\,n\}$$
for the minimum $b$-symbol distance of the linear code $C$.

\begin{Definition}(\cite{LP2})\label{bmds}
An $[n,k]_q$-linear code $C$ with $$d_{b}(C)= \min\{n-k+b,n\}$$ is called a $b$-symbol maximum distance separable ($\,b$-symbol MDS) code. An $[n,k]_q$-linear code $C$ with $$d_{b}(C)= \min\{n-k+b,n\}-1$$ is called a $b$-symbol almost maximum distance separable ($\,b$-symbol AMDS) code.
\end{Definition}

In \cite{LP2}, a result on $b$-symbol MDS codes for different $b$ are given.
\begin{lem}(Theorem 3.11 of \cite{LP2})\label{bmdss}
Let $1\leq b_1\leq b_2\leq n$. If $C$ is a $b_1$-symbol MDS code, then $C$ is a $b_2$-symbol MDS code.
\end{lem}

When we study the $b$-symbol support and the $b$-symbol weight of codewords, we need the following definition.
\begin{Definition}\label{hole}
For any subset $J$ of $\mathbb{Z}_n=\{1,2,\cdots,n\}$, a hole $H$ of $J$ with size $|H|=h$ is defined as a nonempty set such that $H=\{a+1,a+2,\cdots,a+h\}\subseteq \mathbb{Z}_n\backslash J$ and $a,a+h+1\in J$. We denote the set of all the holes of $J$ by $\mathbb{H}(J)$.
\end{Definition}

If $|\mathbb{H}(J)|\leq 1$, we say $J$ is a {\it successive subset} of $\mathbb{Z}_n$.
 The following lemma gives the relationship between the Hamming weight and the $b$-symbol weight of any $\mathbf{x}\in\mathbb{F}_q^{n}$.

\begin{lem}\label{relationship-1}
\begin{description}
  \item[(a)] Let  $\mathbf{x}\in\mathbb{F}_q^{n}$. Then $$w_{b}(\mathbf{x})=w_{1}(\mathbf{x})+\sum_{H\in \mathbb{H}(\chi_1(\mathbf{x})),\,|H|\leq b-1}|H|+\sum_{H\in \mathbb{H}(\chi_1(\mathbf{x})),\,|H|\ge b}(b-1).$$
  \item[(b)] Let $C$ be an $[n,k]_q$-linear code. If there exists a codeword $\mathbf{c}\in C$ such that $$w_1(\mathbf{c})=d_1(C)$$ and $\chi_1(\mathbf{c})$ is a successive subset of $\mathbb{Z}_n$, then $$d_b(C)=\min\{d_1(C)+b-1,n\}$$ for every $1\leq b\leq n.$
\end{description}
\end{lem}

\begin{proof}
\textbf{(a)} It was proved in Lemma 3.1 of \cite{LP2}.

\textbf{(b)} It is easy to be proved using the definition of a successive subset and \textbf{(a)}.
\end{proof}

In the following, we introduce the definition of the matrix product code $[C_{1},\cdots, C_{M}]\cdot A$, where $A=(a_{i,j})_{M\times N}$ is an $M\times N$ matrix of rank $M$ over $\mathbb{F}_{q}$, and $C_{1},\cdots,C_{M}$ are codes of length $n$ over $\mathbb{F}_{q}$. We assume $M\leq N$ in this paper.

\begin{Definition}\label{MP}
The matrix product code $[C_{1},\cdots, C_{M}]\cdot A$ is the set of all the $1\times nN$ row vectors $\big{[}{\bf c}_{1}, \cdots, {\bf c}_{M}\big{]}\cdot A = \big{[}\sum_{\ell=1}^{M}{\bf c}_{\ell}a_{\ell, 1},\cdots, \sum_{\ell=1}^{M}{\bf c}_{\ell}a_{\ell ,N}\big{]}$, where ${\bf c}_{\ell}\in C_{\ell}$ is a $1 \times n$ row vector for $1\leq \ell\leq M$.
\end{Definition}

In fact, assuming $\mathbf{c}_\ell=(c_{1,\ell},c_{2,\ell},\cdots,c_{n,\ell})\in C_{\ell}$ for $1\leq \ell\leq M$, we have
$$\mathbf{c}=\big{[}{\bf c}_{1}, \cdots, {\bf c}_{M}\big{]}\cdot A =(b_{1,1},\cdots, b_{n,1},b_{1,2},\cdots, b_{n,2},\cdots, b_{1,N},\cdots ,b_{n,N}),$$
where $b_{i,j}=\sum_{\ell=1}^{M} c_{i,\ell}a_{\ell,j}$. Hence there exists an injective map $\Delta$ from the matrix product code $C=[C_{1},\cdots, C_{M}]\cdot A$ to the set of $n\times N$ matrices over $\mathbb{F}_{q}$ such that
$$\Delta(\mathbf{c})=\left(\begin{array}{ccc}
b_{1,1}  &\cdots  &b_{1,N} \\
\vdots  & \ddots &\vdots \\
b_{n,1} & \cdots &b_{n,N}
\end{array}\right)$$
for any $\mathbf{c}=(b_{1,1},\cdots, b_{n,1},b_{1,2},\cdots, b_{n,2},\cdots, b_{1,N},\cdots ,b_{n,N})\in C$.

\begin{Proposition}[\cite{BN} or Proposition 5.5 of \cite{LP1}]\label{ge}
Let $C_{1},\cdots,C_{M}$ be linear codes of length $n$ over $\mathbb{F}_{q}$ and let $A=(a_{i,j})_{M\times N}$ be an $M\times N$ matrix of rank $M$ over $\mathbb{F}_{q}$. Let $C=[C_{1},\cdots, C_{M}]\cdot A$, and let $G_{\ell}$ be a generator matrix of $C_{\ell}$ for all $1\leq \ell\leq M$. Then $$G=\left(\begin{array}{ccc}
a_{1,1}G_{1}  &\cdots  &a_{1,N}G_{1} \\
\vdots  & \ddots &\vdots \\
a_{M,1}G_{M} & \cdots &a_{M,N}G_{M}
\end{array}\right)$$ is a generator matrix of $C$.
\end{Proposition}

A lower bound for the minimum Hamming weight of the matrix product code $[C_{1},\cdots, C_{M}]\cdot A$ was given in \cite{BN}, where $A$ is a nonsingular by column matrix defined as follows.

For $1\leq t\leq M$, let $A_t$ be the matrix consisting of the first $t$ rows of $A$. For $1\leq j_1 <\cdots <j_t\leq N$, we use $A_t(j_1,\cdots, j_t)$ to denote the $t\times t$ submatrix consisting of columns $j_1,\cdots, j_t$ of $A_t$.

\begin{Definition}\label{NSC}
A matrix $A$ is called nonsingular by column (NSC), if $A_t(j_1,\cdots, j_t)$ is nonsingular for each $1\leq t\leq M$ and $1\leq j_1 <\cdots <j_t\leq N$.
\end{Definition}

Let $A$ be an NSC matrix. It is easy to verify that the linear code generated by the first $t$ rows of $A$ is a $1$-symbol MDS code for every $1\leq t\leq M$.

\section{Bounds for $b$-symbol weights of matrix product codes}

In this section, we always assume $1\le b\le n$. We will give a lower bound for the minimum $b$-symbol distance of matrix product codes in this section. The following lemma gives a characterization of the $b$-symbol weight of a vector related to a matrix.
\begin{lem}\label{lem31}
Let $J$ be a subset of $\mathbb{Z}_n$ with $|J|=g>0$.
Let $$\mathbf{c}=(b_{1,1},\cdots, b_{n,1},b_{1,2},\cdots, b_{n,2},\cdots, b_{1,N},\cdots ,b_{n,N})\in \mathbb{F}_q^{nN}$$
 and $$\Delta(\mathbf{c})=\left(\begin{array}{ccc}
b_{1,1}  &\cdots  &b_{1,N} \\
\vdots  & \ddots &\vdots \\
b_{n,1} & \cdots &b_{n,N}
\end{array}\right).$$
If the $i$th row of the matrix $\Delta(\mathbf{c})$ has exactly $\ell$ nonzero elements for each $i\in J$ and the other rows of $\Delta(\mathbf{c})$ are zeros, then $$w_b(\mathbf{c})\geq\ell(g+\sum_{H\in \mathbb{H}(J),\,|H|\leq b-1}|H|+\sum_{H\in \mathbb{H}(J),\,|H|\ge b}(b-1)).$$
\end{lem}

\begin{proof}
Let $\chi_1(\mathbf{c})\subseteq \mathbb{Z}_{nN}$ be the Hamming support of $\mathbf{c}$.
We claim that every hole $H\in\mathbb{H}(J)$ will  yield $\ell$ different holes $\tilde{H}_{1},\,\tilde{H}_{2},\,\cdots,\,\tilde{H}_{\ell}\in \mathbb{H}(\chi_1(\mathbf{c}))$ such that $|H|\leq |\tilde{H}_{i}|$ for $1\leq i\leq \ell$, where $J\subseteq\mathbb{Z}_n$ and $\chi_1(\mathbf{c})\subseteq \mathbb{Z}_{nN}$.

Let $H=\{a+1,a+2,\cdots,a+h\}\subseteq \mathbb{Z}_n\backslash J$ be a hole of $J$ with $|H|=h$.
Since $a,\,a+h+1\in J$, there exist exactly $\ell$ nonzero elements $b_{a,j_{1}},\,b_{a,j_{2}},\,\cdots,\,b_{a,j_{\ell}}$ for $1\leq j_{1}<j_{2}<\cdots<j_{\ell}\leq N$.
Then the  submatrix
$$\left(\begin{array}{ccc}
b_{a+1,j_{1}}  &\cdots  &b_{a+1,j_{\ell}} \\
\vdots  & \ddots &\vdots \\
b_{a+h,j_{1}} & \cdots &b_{a+h,j_{\ell}}
\end{array}\right)$$
 of $\Delta(\mathbf{c})$ is the zero matrix.
Note that  $$a+n(j_i-1)\in \chi_1(\mathbf{c})$$ for $1\leq i\leq \ell$. There exist $\ell$ different holes $$\tilde{H}_{1},\,\tilde{H}_{2},\,\cdots,\,\tilde{H}_{\ell}\in \mathbb{H}(\chi_1(\mathbf{c}))$$ such that $$\{a+1+n(j_{i}-1),\,a+2+n(j_{i}-1),\,\cdots,\,a+h+n(j_{i}-1)\}\subseteq \tilde{H}_{i}$$ and $|H|\leq |\tilde{H}_{i}|$ for $1\leq i\leq \ell$.

From the discussion above, we also know that two different holes in $\mathbb{H}(J)$ will yield $2\ell$ different holes in $\mathbb{H}(\chi_1(\mathbf{c}))$.
By Lemma \ref{relationship-1}, we have $$w_b(\mathbf{c})=w_{1}(\mathbf{c})+\sum_{H\in \mathbb{H}(\chi_1(\mathbf{c})),\,|H|\leq b-1}|H|+\sum_{H\in \mathbb{H}(\chi_1(\mathbf{c})),\,|H|\ge b}(b-1)$$
$$\geq\ell g+\ell\sum_{H\in \mathbb{H}(J),\,|H|\leq b-1}|H|+\ell\sum_{H\in \mathbb{H}(J),\,|H|\ge b}(b-1).$$
\end{proof}

In the following theorem, we give a lower bound for the minimum $b$-symbol distance of matrix product codes.

\begin{Theorem}\label{t3.1}
Let $C_{1},\cdots,C_{M}$ be codes of length $n$ over $\mathbb{F}_{q}$ and let $A$ be an $M\times N$ matrix of rank $M$. Let $C=[C_{1},\cdots, C_{M}]\cdot A$. Then the minimum $b$-symbol distance $d_b(C)$ of $C$ satisfies the following:

\begin{description}
  \item[(a)] Assume $t_i$ is the minimum Hamming distance of a linear code generated by the first $i$ rows of $A$.  Then
      \begin{equation}\label{3.1}
d_b(C)\geq\min\{t_id_b(C_i)\,|\,i=1,2,\cdots,M\}.
\end{equation}

  \item[(b)] Assume $s_i$ is the minimum Hamming distance of a linear code generated by the last $i$ rows of $A$. Then
      \begin{equation}\label{3.2}
d_b(C)\geq\min\{s_{M-i+1}d_b(C_i)\,|\,i=1,2,\cdots,M\}.
\end{equation}
  \item[(c)] In particular, if $A$ is an NSC matrix, then
\begin{equation}\label{3.3}
d_b(C)\geq\min\{(N-i+1)d_b(C_i)\,|\,i=1,2,\cdots,M\}.
\end{equation}
\end{description}
\end{Theorem}

\begin{proof}
\textbf{(a)} Assume $\mathbf{c}$ is a codeword in $C$ such that $w_b(\mathbf{c})=d_b(C)$.
Then
$$\mathbf{c}=\big{[}{\bf c}_{1}, \cdots, {\bf c}_{M}\big{]}\cdot A =\big{[}\sum_{\ell=1}^{M}{\bf c}_{\ell}a_{\ell, 1},\cdots, \sum_{\ell=1}^{M}{\bf c}_{\ell}a_{\ell ,N}\big{]},$$
for some  $\mathbf{c}_\ell=(c_{1,\ell},c_{2,\ell},\cdots,c_{n,\ell})\in C_{\ell}$ and $1\leq \ell\leq M$.
Let $$\mu=\max\{1\leq \ell\leq M\,|\,{\bf c}_{\ell}\neq 0\}.$$
Then we have $$\mathbf{c}=(b_{1,1},\cdots, b_{n,1},b_{1,2},\cdots, b_{n,2},\cdots, b_{1,N},\cdots ,b_{n,N}),$$
where $b_{i,j}=\sum_{\ell=1}^{\mu} c_{i,\ell}a_{\ell,j}$ and
 $$\Delta(\mathbf{c})=\left(\begin{array}{ccc}
b_{1,1}  &\cdots  &b_{1,N} \\
\vdots  & \ddots &\vdots \\
b_{n,1} & \cdots &b_{n,N}
\end{array}\right)=\left(\begin{array}{ccc}
c_{1,1}  &\cdots  &c_{1,M} \\
\vdots  & \ddots &\vdots \\
c_{n,1} & \cdots &c_{n,M}
\end{array}\right)\left(\begin{array}{ccc}
a_{1,1}  &\cdots  &a_{1,N} \\
\vdots  & \ddots &\vdots \\
a_{M,1} & \cdots &a_{M,N}
\end{array}\right)   .$$
Since $\mathbf{c}_{\ell}=\mathbf{0}$ for any ${\mu}<\ell\leq M$, we know that every row of the matrix $\Delta(\mathbf{c})$ is a linear combination of the first ${\mu}$ rows of the matrix $A$. Since the rank of $A$ is $M$, we have that the $i$th row of $\Delta(\mathbf{c})$ is the zero row if and only if $$(c_{i,1},\,c_{i,2},\,\cdots,\,c_{i,M})=\mathbf{0}.$$
Hence there exist at least $w_1(\mathbf{c}_{\mu})$ nonzero rows of $\Delta(\mathbf{c})$, and the Hamming weight of each nonzero row of $\Delta(\mathbf{c})$ is great than or equal to $t_{\mu}$.

If we change some nonzero elements of $\mathbf{c}$ into zeros, the $b$-symbol weight of $\mathbf{c}$ will decrease. Hence we can assume $\mathbf{c}$ satisfies that the $i$th row of the matrix $\Delta(\mathbf{c})$ has exactly $t_{\mu}$ nonzero elements for each $i\in J$, and other rows of the matrix $\Delta(\mathbf{c})$ are zero rows, where $J=\chi_1(\mathbf{c}_{\mu})$.
By Lemmas \ref{relationship-1} and \ref{lem31}, we get $$w_b(\mathbf{c})\geq t_{\mu}(g+\sum_{H\in \mathbb{H}(J),\,|H|\leq b-1}|H|+\sum_{H\in \mathbb{H}(J),\,|H|\ge b}(b-1))=t_{\mu} w_b(\mathbf{c}_{\mu})\ge t_{\mu} d_b(C_{\mu}),$$
where $g=|J|=|\chi_1(\mathbf{c}_{\mu})|=w_1(\mathbf{c}_{\mu})$.
Hence $$d_b(C)=w_b(\mathbf{c})\ge \min\{t_id_b(C_i)\,|\,i=1,2,\cdots,M\}.$$

\textbf{(b)} Assume $\mathbf{c}$ is a codeword in $C$ such that $w_b(\mathbf{c})=d_b(C)$. Then
$$\mathbf{c}=\big{[}{\bf c}_{1}, \cdots, {\bf c}_{M}\big{]}\cdot A =\big{[}\sum_{\ell=1}^{M}{\bf c}_{\ell}a_{\ell, 1},\cdots, \sum_{\ell=1}^{M}{\bf c}_{\ell}a_{\ell ,N}\big{]},$$
for some  $\mathbf{c}_\ell=(c_{1,\ell},c_{2,\ell},\cdots,c_{n,\ell})\in C_{\ell}$ and $1\leq \ell\leq M$.
Let $$\nu=\min\{1\leq \ell\leq M\,|\,{\bf c}_{\ell}\neq 0\}.$$
Then we have $$\mathbf{c}=(b_{1,1},\cdots, b_{n,1},b_{1,2},\cdots, b_{n,2},\cdots, b_{1,N},\cdots ,b_{n,N}),$$
 where $b_{i,j}=\sum_{\ell=\nu}^{M} c_{i,\ell}a_{\ell,j}$ and
 $$\Delta(\mathbf{c})=\left(\begin{array}{ccc}
b_{1,1}  &\cdots  &b_{1,N} \\
\vdots  & \ddots &\vdots \\
b_{n,1} & \cdots &b_{n,N}
\end{array}\right)=\left(\begin{array}{ccc}
c_{1,1}  &\cdots  &c_{1,M} \\
\vdots  & \ddots &\vdots \\
c_{n,1} & \cdots &c_{n,M}
\end{array}\right)\left(\begin{array}{ccc}
a_{1,1}  &\cdots  &a_{1,N} \\
\vdots  & \ddots &\vdots \\
a_{M,1} & \cdots &a_{M,N}
\end{array}\right)   .$$

Since $\mathbf{c}_{\ell}=\mathbf{0}$ for any $1\leq \ell<\nu$, we know that every row of the matrix $\Delta(\mathbf{c})$ is a linear combination of the last $M-\nu+1$ rows of the matrix $A$. Since the rank of $A$ is $M$, we have that the $i$th row of $\Delta(\mathbf{c})$ is the zero row if and only if $$(c_{i,1},\,c_{i,2},\,\cdots,\,c_{i,M})=\mathbf{0}.$$
Hence there exist at least $w_1(\mathbf{c}_{\nu})$ nonzero rows of $\Delta(\mathbf{c})$, and the Hamming weight of each nonzero row of $\Delta(\mathbf{c})$ is greater than or equal to $s_{M-\nu+1}$.

If we change some nonzero elements of $\mathbf{c}$ into zeros, the $b$-symbol weight of $\mathbf{c}$ will decrease. Hence we can assume $\mathbf{c}$ satisfies that the $i$th row of the matrix $\Delta(\mathbf{c})$ has exactly $s_{M-\nu+1}$ nonzero elements for each $i\in J$, and other rows of the matrix $\Delta(\mathbf{c})$ are zero rows, where $J=\chi_1(\mathbf{c}_{\nu})$.
By Lemmas \ref{relationship-1} and \ref{lem31}, we get $$w_b(\mathbf{c})\geq s_{M-\nu+1}(g+\sum_{H\in \mathbb{H}(J),\,|H|\leq b-1}|H|+\sum_{H\in \mathbb{H}(J),\,|H|\ge b}(b-1))=s_{M-\nu+1} w_b(\mathbf{c}_{\nu})\ge s_{M-\nu+1} d_b(C_{\nu}),$$
where $g=|J|=|\chi_1(\mathbf{c}_{\nu})|=w_1(\mathbf{c}_{\nu})$.
Hence $$d_b(C)=w_b(\mathbf{c})\ge \min\{s_{M-i+1}d_b(C_i)\,|\,i=1,2,\cdots,M\}.$$

\textbf{(c)}  If $A$ is an NSC matrix, we know that the linear code generated by the first $t$ rows of $A$ is a $1$-symbol MDS code for every $1\leq t\leq M$. Hence $$t_i=N-i+1$$ in statement \textbf{(a)}  for $1\leq i\leq M$.
\end{proof}

For convenience, we let $d^*=\min\{(N-i+1)d_b(C_i)\,|\,i=1,2,\cdots,M\}.$
The matrix $ A=(a_{i,j})_{M\times N}$ is called an {\it upper triangular matrix} if $a_{i,j}=0$ for $i>j$.
We have the following corollary.

\begin{Corollary}\label{t3.2}
Assume the notation is as given above.
Let $C=[C_{1},\cdots, C_{M}]\cdot A$ be a matrix product code where $A=(a_{i,j})_{M\times N}$ is an $M\times N$  upper triangular NSC matrix. Then
\begin{description}
  \item[(a)] The minimum $b$-symbol distance $d_b(C)$ of $C$ satisfies
$$d_b(C)\leq\min\{Nd_b(C_1),\,(N-i+1)d_b(C_i)+b-1\,|\,i=2,3,\cdots,M\}.$$
  \item[(b)] If $d^*=Nd_b(C_1)$, then $d_b(C)=d^*$.
  \item[(c)] Let  $2\leq i_0\leq M$ such that $(N-i_{0}+1)d_b(C_{i_0})=d^*.$
Assume there exists a codeword $\mathbf{c}_{i_0}\in C_{i_0}$ and $H\in \mathbb{H}(\chi_1(\mathbf{c}_{i_0}))$ such that $w_b(\mathbf{c}_{i_0})=d_b(C_{i_0})$, $|H|\ge b-1$ and $\{1,n\}\bigcap H\neq \emptyset$.
Then $$d_b(C)=d^*.$$
\end{description}

\end{Corollary}

\begin{proof}
\textbf{(a)} Let $\mathbf{c}_i\in C_i$ such that $w_b(\mathbf{c}_i)=d_b(C_i)$ for $1\leq i \leq M$.
Assume $$\Bar{\mathbf{c}}_i=\big{[}\mathbf{0}, \cdots, \mathbf{0},{\bf c}_{i},\mathbf{0},\cdots, \mathbf{0}\big{]}\cdot A\in C$$ for $1\leq i\leq M$, then $$\Bar{\mathbf{c}}_1=\big{[}a_{1,1}\mathbf{c}_1,a_{1,2}\mathbf{c}_1,\cdots,a_{1,N}\mathbf{c}_1\big{]}$$
and for $2\leq i\leq M$,
$$\Bar{\mathbf{c}}_i=\big{[}\mathbf{0},\cdots,\mathbf{0},a_{i,i}\mathbf{c}_i,a_{i,i+1}\mathbf{c}_i,\cdots,a_{i,N}\mathbf{c}_i, \big{]}.$$
Since a triangular NSC matrix has exactly $(i-1)$ zeros in row $i$ by Proposition 3.5 of \cite{BN}, we have $w_b(\Bar{\mathbf{c}}_1)=N w_b(\mathbf{c}_1)=Nd_b(C_1)$ and for $2\leq i\leq M$,
$$ w_b(\Bar{\mathbf{c}}_i)\leq (N-i+1)w_b(\mathbf{c}_i)+b-1=(N-i+1)d_b(C_i)+b-1.$$
Hence $$d_b(C)\leq\min\{Nd_b(C_1),\,(N-i+1)d_b(C_i)+b-1\,|\,i=2,3,\cdots,M\}.$$

\textbf{(b)} By Theorem~\ref{t3.1} and statement \textbf{(a)}, we have $$Nd_b(C_1)=d^*=d_b(C). $$

\textbf{(c)} Assume there exists a codeword $\mathbf{c}_{i_0}\in C_{i_0}$ and $H\in \mathbb{H}(\chi_1(\mathbf{c}_{i_0}))$ such that $$w_b(\mathbf{c}_{i_0})=d_b(C_{i_0}),$$ $|H|\ge b-1$ and $1\in H$.
Since $$\Bar{\mathbf{c}}_{i_0}=\big{[}\mathbf{0},\cdots,\mathbf{0},a_{i_0,i_0}\mathbf{c}_{i_0},a_{i_0,i_0+1}\mathbf{c}_{i_0},\cdots,a_{i_0,N}\mathbf{c}_{i_0} \big{]},$$ we have $$w_b(\Bar{\mathbf{c}}_{i_0})= (N-i_0+1)w_b(\mathbf{c}_{i_0})=(N-i_0+1)d_b(C_{i_0}).$$
Hence we have statement \textbf{(c)}.
\end{proof}

\begin{Remark}
\begin{description}
  \item[(a)] Theorem 3.7 of \cite{BN} can be easily obtained from Theorem~\ref{t3.1} and Corollary~\ref{t3.2}, when $b=1$.
Furthermore, Theorem 4 of \cite{LEL} is a special case of Theorem~\ref{t3.1} when $b=2$.
  \item[(b)] All the codes $C_{1},\cdots,C_{M}$ in Theorem~\ref{t3.1} and Corollary~\ref{t3.2} are not necessarily linear codes.
\end{description}

\end{Remark}

\section{The $b$-symbol distances of Reed-Muller codes}
In this section,  we determine the minimum $b$-symbol distance of Reed-Muller codes.
Assume $$\mathbb{F}_q = \{\alpha_1=0,\alpha_2,\cdots,\alpha_q \}.$$
In particular, when $q$ is a prime number, we can take $$\alpha_i=i-1$$ for $1\leq i\leq q$.
  We define a partial order ``$\leq $" on $\mathbb{F}_q$ as following:
 $$\alpha_1\leq \alpha_2\leq \cdots\leq \alpha_q,$$
i.e., $\alpha_i<\alpha_j$ means $i<j$.

The partial order ``$\leq $" on $\mathbb{F}_q$ can be easily extended to $\mathbb{F}_q^m$ by using lexicographical order. Let $\mathbf{x}=(x_1,\cdots,x_m),\,\mathbf{y}=(y_1,\cdots,y_m)\in \mathbb{F}_q^m$ be any two vectors. Then $\mathbf{x}<\mathbf{y}$
means there exists $1\leq j \leq m $ such that $x_i=y_i$ for $1\leq i\leq j-1$ and $x_j<y_j$.

For example, let $q=2$ and $m=3$. Assume $\mathbf{x}_i$ is the $i$th column of the following matrix
$$\left(\begin{array}{cccccccc}
0  &0  &0 &0 &1 &1 &1 &1 \\
0  &0  &1 &1 &0 &0 &1 &1 \\
0  &1  &0 &1 &0 &1 &0 &1
\end{array}\right) .$$
Then we have $\mathbf{x}^T_1<\mathbf{x}^T_2<\cdots<\mathbf{x}^T_8.$

Let $\mathbb{F}_q[X_1,\cdots,X_m]$ be the polynomial ring in $m$ variables over $\mathbb{F}_q$ and $$\mathbb{F}_q^{\leq r}[X_1,\cdots,X_m]=\{f\in\mathbb{F}_q[X_1,\cdots,X_m]\,|\, \deg{f}\leq r \}$$ for $m\ge 1$ and $ r\ge 0$.
 The {\it Reed-Muller code} $RM_q(r,m)$ over $\mathbb{F}_q$ is defined as:
$$RM_q(r,m)=\{(f(P_1),f(P_2),\cdots,f(P_{q^m}))\,|\,f\in \mathbb{F}_q^{\leq r}[X_1,\cdots,X_m]\},$$
where $\mathbb{F}_q^m=\{P_1,P_2,\cdots,P_{q^m}\}$ such that $P_1<P_2<\cdots<P_{q^m}$.
By Corollary 1.2 of \cite{DG}, we know $$RM_q(r,m)=\mathbb{F}_q^{q^m},$$ when $  r\ge m(q-1)$.
For convenience, we assume $RM_q(r,m)=\{\mathbf{0}\}$ when $r<0$.

Recall $\mathbb{F}_q = \{\alpha_1,\alpha_2,\cdots,\alpha_q \}$.
For $2\leq i\leq q$ and $1\leq  j\leq q$, we define $$\binom{\alpha_j}{\alpha_i}=\frac{(\alpha_j-\alpha_1)(\alpha_j-\alpha_2)\cdots(\alpha_j-\alpha_{i-1})}{(\alpha_i-\alpha_1)(\alpha_i-\alpha_2)\cdots(\alpha_i-\alpha_{i-1})}.$$
For $i=1$, we define $\binom{\alpha_j}{\alpha_i}=1$.
We know that $\binom{\alpha_j}{\alpha_i}=0$ if and only if $1\leq j\leq i-1$.
Then we can define the matrix $$G_q=\left(\begin{array}{cccc}
\binom{\alpha_1}{\alpha_1}  &\binom{\alpha_2}{\alpha_1}  &\cdots &\binom{\alpha_q}{\alpha_1}\\
\binom{\alpha_1}{\alpha_2} & \binom{\alpha_2}{\alpha_2}  &\cdots &\binom{\alpha_q}{\alpha_2}\\
\vdots & \vdots & \ddots&\vdots\\
\binom{\alpha_1}{\alpha_q} & \binom{\alpha_1}{\alpha_q} &\cdots &\binom{\alpha_q}{\alpha_q}
\end{array}\right) . $$

It is obvious that $G_q$ is upper triangular and has $1$ on its leading diagonal.

\begin{Proposition}\label{4.1aa}[Theorem 5.6 of \cite{BN}]
Assume the notation is as given above. The Reed-Muller code $RM_q(r,m)$ can be iteratively defined by $$RM_q(r,0)=\left\{ \begin{array}{ll}
\{\mathbf{0}\},  & \textrm{if $r<0 ;$}\\
\mathbb{F}_q,  & \textrm{if $r\ge 0 ,$}
\end{array} \right.$$
and for $m\ge 1$, $$RM_q(r,m)=[RM_q(r,m-1),\,RM_q(r-1,m-1),\,\cdots,\,RM_q(r-q+1,m-1)]\cdot G_q.$$
\end{Proposition}

The following result gives the minimum Hamming distance of Reed-Muller codes.
\begin{Proposition}\label{4.1}[Theorem 2.6.2 of \cite{DG} and \cite{KLP}]
Assume the notation is as given above. Let $r=t(q-1)+s$, where $0\leq s<q-1$ and $t\ge 0$. Then $$d_1(RM_q(r,m))=(q-s)q^{m-t-1}.$$
\end{Proposition}

\begin{lem}\label{dee}
Assume the notation is given above. Suppose $r\ge q-1$, then $$d_1(RM_q(r,m))=d_1(RM_q(r-q+1,m-1)).$$
\end{lem}
\begin{proof}
Assume $r=t(q-1)+s$, where $0\leq s<q-1$ and $t\ge 1$, then $$r-q+1=(t-1)(q-1)+s.$$ By Proposition~\ref{4.1}, we have $$d_1(RM_q(r,m))=(q-s)q^{m-t-1}=d_1(RM_q(r-q+1,m-1)).$$
\end{proof}
\begin{lem}\label{succ1}
Assume $\mathbb{F}_q = \{\alpha_1=0,\alpha_2,\cdots,\alpha_q \}$.
Let $\mathbf{x}_i=(\alpha_1^i,\alpha_2^i,\cdots,\alpha_q^i)\in \mathbb{F}_q^q$ for $1\leq i \leq q-1$. For $1\leq r\leq q-1$, then there exist $a_i\in \mathbb{F}_q$ for $1\leq i\leq r$ such that $$\sum_{i=1}^{r}a_i\mathbf{x}_i=(y_1,y_2,\cdots,y_q)\neq\mathbf{0} $$
and $y_j=0$ for $1\leq j\leq r$.

\end{lem}

\begin{proof}
Note that $\alpha_1=0$, it is easy to prove this lemma using basic linear algebra.

\end{proof}

\begin{lem}\label{succ}
Suppose $m\ge 1$ and $0\leq r<q-1$. Then there exists a codeword $\mathbf{c}\in RM_q(r,m)$ such that $$w_1(\mathbf{c})=d_1(RM_q(r,m))$$
and $\chi_1(\mathbf{c})$ is a successive subset of $\mathbb{Z}_{q^{m}}$ such that $1\notin \chi_1(\mathbf{c})$ unless $\chi_1(\mathbf{c})=  \mathbb{Z}_{q^{m}}$.
\end{lem}

\begin{proof}
Assume $\mathbb{F}_q = \{\alpha_1=0,\alpha_2,\cdots,\alpha_q \}$.
When $r=0$, we know that $$RM_q(0,m)=\{\alpha\mathbf{1}\,|\,\alpha\in\mathbb{F}_q\},$$
$w_1(\mathbf{1})=d_1(RM_q(0,m))$  and  $\chi_1(\mathbf{1})=  \mathbb{Z}_{q^{m}},$
 where $\mathbf{1}\in \mathbb{F}_q^{q^m}.$

Assume $r\ge 1$.
Let $$f_i(X_1,\cdots,X_m)=X^i_1\in \mathbb{F}_q[X_1,\cdots,X_m]$$ for $0\leq i\leq  q-1$.
Then $$(f_i(P_1),f_i(P_2),\cdots,f_i(P_{q^m}))=(\alpha_1^i\mathbf{1},\,\alpha_2^i\mathbf{1},\,\cdots ,\,\alpha_q^i\mathbf{1})\in RM_q(r,m),$$
where $\mathbf{1}\in \mathbb{F}_q^{q^{m-1}}.$
By Lemma~\ref{succ1}, there exist $a_i\in \mathbb{F}_q$ for $1\leq i\leq r$ such that
 $$\sum_{i=1}^{r}a_i(f_i(P_1),f_i(P_2),\cdots,f_i(P_{q^m}))=(y_1\mathbf{1},y_2\mathbf{1},\cdots,y_q\mathbf{1})\neq \mathbf{0},$$
where $y_j=0$ for $1\leq j\leq r$ and $\mathbf{0}\in \mathbb{F}_q^{q^m}.$
Let $$f(X_1,\cdots,X_m)=\sum^r_{i=1} a_i X^i_1\in \mathbb{F}_q^{\leq r}[X_1,\cdots,X_m],$$
then we have $$(f(P_1),\cdots,f(P_{q^m}))=\sum_{i=1}^{r}a_i(f_i(P_1),\cdots,f_i(P_{q^m}))=(y_1\mathbf{1},\cdots,y_q\mathbf{1})\in RM_q(r,m).$$
Since $y_j=0$ for $1\leq j\leq r$, we have
 $$d_1(RM_q(r,m))\leq w_1((f(P_1),f(P_2),\cdots,f(P_{q^m})))\leq (q-r)q^{m-1}.$$
 By Proposition~\ref{4.1}, we know that $$w_1((f(P_1),f(P_2),\cdots,f(P_{q^m})))= (q-r)q^{m-1}$$ and $$y_j\neq 0,$$ for $r+1\leq j\leq q-1$.
 Hence $\mathbf{c}=(f(P_1),f(P_2),\cdots,f(P_{q^m}))$ is what we desire.
\end{proof}

\begin{Theorem}\label{weight}
Assume the notation is as given above. Let $m\ge 1$ and $r=t(q-1)+s$, where $0\leq s<q-1$ and $t\ge 0$.
Then $$d_b(RM_q(r,m))=\min\{(q-s)q^{m-t-1}+b-1,\,q^m\}.$$

\end{Theorem}

\begin{proof}
First we claim that there exists a codeword $\mathbf{c}\in RM_q(r,m)$ such that $$w_1(\mathbf{c})=d_1(RM_q(r,m))$$ and $\chi_1(\mathbf{c})$ is a successive subset of $\mathbb{Z}_{q^m}$ such that $1\notin \chi_1(\mathbf{c})$ unless $\chi_1(\mathbf{c})=\mathbb{Z}_{q^m}$.
We prove this by induction on $m$.

Assume $m=1$. When $r\ge m(q-1)=q-1$, we have $RM_q(r,m)=\mathbb{F}_q^{q^m}.$
When $0\leq r< q-1$, we can get what we need by Lemma~\ref{succ}.

Now assume $m\ge 2$. By Lemma~\ref{succ}, we can assume that $r\ge q-1$.
By Proposition \ref{4.1aa}, we know that
$$RM_q(r,m)=[RM_q(r,m-1),\,RM_q(r-1,m-1),\,\cdots,\,RM_q(r-q+1,m-1)]\cdot \left(\begin{array}{cccc}
1  &a_{1,2}&\cdots&a_{1,q}   \\
0 &1&\cdots&a_{2,q}   \\
\vdots  &\vdots&\ddots&\vdots   \\
0  & 0&0&1
\end{array}\right) .$$

Let $G(r-i,m-1)$ be a generator matrix of $RM_q(r-i,m-1)$ for $0\leq i \leq q-1$.  By Proposition~\ref{ge},  $$G(r,m)=\left(\begin{array}{cccc}
G(r,m-1)  &a_{1,2}G(r,m-1)&\cdots&a_{1,q}G(r,m-1)   \\
0 &G(r-1,m-1)&\cdots&a_{2,q}G(r-1,m-1)   \\
\vdots  &\vdots&\ddots&\vdots   \\
0  & 0&0&G(r-q+1,m-1)
\end{array}\right) $$
is a generator matrix of $RM_q(r,m)$.
By Lemma~\ref{dee}, $$d_1(RM_q(r-q+1,m-1))=d_1(RM_q(r,m)).$$
By induction, there exists $\mathbf{c}\in RM_q(r-q+1,m-1)$ such that $$w_1(\mathbf{c})=d_1(RM_q(r-q+1,m-1))=d_1(RM_q(r,m))$$
and $\chi_1(\mathbf{c})$ is a successive subset of $\mathbb{Z}_{q^{m-1}}$ such that $1\notin \chi_1(\mathbf{c})$ unless $\chi_1(\mathbf{c})=  \mathbb{Z}_{q^{m-1}}$.
Note that $G(r,m)$ is the generator matrix of $RM_q(r,m)$, we have  $$\tilde{\mathbf{c}}=[\mathbf{0},\,,\cdots,\,\mathbf{0},\,\,\mathbf{c}]\in RM_q(r,m),$$
then $$w_1(\tilde{\mathbf{c}})=d_1(RM_q(r-q+1,m-1))=d_1(RM_q(r,m))$$ and $\chi_1(\tilde{\mathbf{c}})$ is a successive subset of $\mathbb{Z}_{q^{m}}$ such that $1\notin \chi_1(\tilde{\mathbf{c}})$ unless $\chi_1(\tilde{\mathbf{c}})=  \mathbb{Z}_{q^{m}}$.

By Lemma~\ref{relationship-1}, we have $$d_b(RM_q(r,m))=\min\{d_1(RM_q(r,m))+b-1,\,q^m\}=\min\{(q-s)q^{m-t-1}+b-1,\,q^m\}.$$
\end{proof}

By Lemma~\ref{bmdss}, we know that if an $[n,k]_q$-linear code $C$ is a $1$-symbol MDS code, then $C$ is a $b$-symbol MDS code for any $1\leq b \leq n.$

\begin{Corollary}\label{4.7}
Assume the notation is as given above.
For $m\ge 1$ and $ r\ge 0$, then $RM_q(r,m)$ is a $b$-symbol MDS code if and only if either of the following holds:
\begin{description}
  \item[(a)] $b\ge  q^m-(q-s)q^{m-t-1}+1;$
  \item[(b)] $RM_q(r,m)$ is a $1$-symbol MDS code.
\end{description}
\end{Corollary}

\begin{proof}
Let $k$ be the dimension of Reed-Muller code $RM_q(r,m)$.
By Theorem~\ref{weight}, we know that $RM_q(r,m)$ is a $b$-symbol MDS code if and only if
\begin{equation}\label{4s}
    \min\{d_1(RM_q(r,m))+b-1,\,q^m\}=\min\{q^m+b-k,\,q^m\}.
\end{equation}
We also know that $$k\leq q^m-d_1(RM_q(r,m))+1$$ by the Singluton bound for
the minimum Hamming distance.
In the following, we prove this theorem by three cases.

{\bf Case 1}. If $b\ge q^m-d_1(RM_q(r,m))+1 =q^m-(q-s)q^{m-t-1}+1$, the equality~(\ref{4s}) holds.

{\bf Case 2}. If $ k\leq b< q^m-d_1(RM_q(r,m))+1 $, the equality~(\ref{4s}) becomes $$d_1(RM_q(r,m))+b-1=q^m$$ which is impossible.

{\bf Case 3}. If $b<k$, the equality~(\ref{4s}) becomes $$d_1(RM_q(r,m))+b-1=q^m+b-k$$ which means $RM_q(r,m)$ is a $1$-symbol MDS code.
\end{proof}

\section{The $b$-symbol distance of linear codes obtained from the $[u+v,u-v]$-construction}
In this section, we always assume $q$ is odd. We give some bounds for the minimum $b$-symbol distance of linear codes obtained from the $[u+v,u-v]$-construction. We use these bounds to construct some $b$-symbol MDS codes.
The linear codes obtained from the $[u+v,u-v]$-construction are matrix product codes $C=[C_{1}, C_{2}]\cdot A$ where $A=\left(\begin{array}{cc}
1  &1  \\
1  & -1
\end{array}\right)$, which are different from Reed-Muller codes.


\begin{Theorem}\label{5.1}
Let $C_{1},C_{2}$ be linear codes of length $n$ over $\mathbb{F}_{q}$ and let $A=\left(\begin{array}{cc}
1  &1  \\
1  & -1
\end{array}\right)$. Let $C=[C_{1}, C_{2}]\cdot A$ and $1\leq b\leq n$. Then

\begin{description}
  \item[(a)] $d_b(C)\ge \min\{2d_b(C_1),\, d_b(C_2)\} $.
  \item[(b)] $d_b(C) \ge \min\{d_b(C_1),\, 2d_b(C_2)\} $.
  \item[(c)] $\min\{d_b(C_1),\, d_b(C_2)\} \leq d_b(C) \leq \min\{2d_b(C_1),\, 2d_b(C_2)\}$.
  \item[(d)] Assume there exists $\mathbf{x}\in C_1\cap C_2$ and $H\in \mathbb{H}(\chi_1(\mathbf{x}))$ such that $$w_b(\mathbf{x})=\min\{d_b(C_1),\, d_b(C_2)\},\,|H|\ge b-1,\,\,\mathrm{and} \,\,\{1,\,n\}\bigcap H\neq \emptyset.$$
 Then $$d_b(C)=\min\{d_b(C_1),\, d_b(C_2)\}.$$
\end{description}

\end{Theorem}

\begin{proof}
\textbf{(a)} By Theorem~\ref{t3.1} \textbf{(a)}, we have $$d_b(C)\ge \min\{2d_b(C_1),\, d_b(C_2)\} .$$

 \textbf{(b)} By Theorem~\ref{t3.1} \textbf{(b)}, we have $$d_b(C)\ge \min\{d_b(C_1),\, 2d_b(C_2)\}.$$

\textbf{(c)} If $d_b(C_1)\leq d_b(C_2)$, we have
$$d_b(C)\geq\min\{d_b(C_1),\, 2d_b(C_2)\}=d_b(C_1)$$
by statement \textbf{(b)}. Let $\mathbf{x}\in C_1$ such that $d_b(C_1)=w_b(\mathbf{x})$, and $$\mathbf{c}=[\mathbf{x},\mathbf{0}]\cdot A=[\mathbf{x}, \mathbf{x}]\in C.$$
Then $$d_b(C)\leq w_b(\mathbf{c})=2w_b(\mathbf{x})=2d_b(C_1).$$

If $d_b(C_1)> d_b(C_2)$, the argument is similar.

\textbf{(d)} Assume $\mathbf{x}\in C_1\cap C_2$ such that $d_b(\mathbf{x})=\min\{d_b(C_1),\, d_b(C_2)\}.$ Since $q$ is odd, we have  $\frac{1}{2}\mathbf{x}\in C_1\cap C_2$
and $$0\neq \mathbf{c}=[\frac{1}{2}\mathbf{x},\frac{1}{2}\mathbf{x}]\cdot A=[\mathbf{x},\mathbf{0}]\in C.$$
Since there exists $H\in \mathbb{H}(\chi_1(\mathbf{x}))$ such that $|H|\ge b-1$ and $\{1,\,n\}\bigcap H\neq \emptyset$,
we have $$w_b(\mathbf{c})=w_b(\mathbf{x})=\min\{d_b(C_1),\, d_b(C_2)\}.$$
By the inequality $\min\{d_b(C_1),\, d_b(C_2)\} \leq d_b(C)\leq w_b(\mathbf{c})$, we get $$d_b(C)=\min\{d_b(C_1),\, d_b(C_2)\}.$$
\end{proof}

If $C_{1}$ and $C_{2}$ are both $b$-symbol MDS codes, we have the following corollary.

\begin{Corollary}\label{5.3ab}
Assume the notation is as given above. Let $C_{1}$ and $C_{2}$ be $b$-symbol MDS codes such that $d_b(C_1)\leq  d_b(C_2)$ and $b\leq \min\{\dim(C_1),\dim(C_2)\}$.
Then $$d_b(C_1)\leq d_{b}(C)\leq d_{b}(C_1)+d_{b}(C_2)-b.$$
\end{Corollary}
\begin{proof}
By assumption, we know that, for $i=1,\,2,$ $$d_b(C_i)= n-\dim(C_i)+b.$$
By $\dim(C)=\dim(C_1)+\dim(C_2)$ and the Singleton-like bound for the minimum $b$-symbol distance $$d_{b}(C)\leq\min\{2n-(\dim(C_1)+\dim(C_2))+b,\,n\},$$
then we have $$d_b(C_1)\leq d_{b}(C)\leq 2n-\dim(C_1)-\dim(C_2)+b=d_{b}(C_1)+d_{b}(C_2)-b.$$

\end{proof}

\begin{Remark}\label{re}
In Corollary~\ref{5.3ab}, if we assume $ d_b(C_2)=b+1$, we know that the linear code $C$ is a $b$-symbol MDS code or a $b$-symbol AMDS code by the proof of Corollary~\ref{5.3ab}.
We can use this method to construct some $b$-symbol AMDS codes. We have the following corollary.
\end{Remark}

\begin{Corollary}\label{1.9a}
Assume the notation is as given above. Let $C_{1}$ and $C_{2}$ be $[n, n-1]_q$-linear $b$-symbol MDS codes. Assume $1\leq b \leq n-2$ and there exists $$\mathbf{x}=(x_1,\cdots, x_n)\in C_{1} \cap C_{2}$$ and $H\in \mathbb{H}(\chi_1(\mathbf{x}))$ such that $w_b(\mathbf{x})= d_b(C_1)$, $|H|\ge b-1$ and $\{1,\,n\}\bigcap H\neq \emptyset$.
Then the $[2n, 2n-2]_q$-linear code $C$ is a $b$-symbol AMDS code.
\end{Corollary}

\begin{proof}
Since $C_{1}$ and $C_{2}$ are $b$-symbol MDS codes and $1\leq b \leq n-2$, we get $$d_b(C_1)=d_b(C_2)=b+1.$$
By Corollary~\ref{5.3ab}, we have $d_b(C_1)\leq d_{b}(C)\leq d_{b}(C_1)+1.$

Since $\mathbf{x}\in C_{1}\bigcap C_{2}$, we have $\mathbf{c}=[\mathbf{x}+\mathbf{x}, \mathbf{x}-\mathbf{x}]=[2\mathbf{x},\mathbf{0}]\in C.$
There exists $1\leq i_0\leq n$ such that $$x_{i_0},x_{i_0+1}\in \mathbb{F}^*_q$$ and $x_j=0$ for $j\neq i_0,i_0+1$ by $w_b(\mathbf{x})= d_b(C_1)=b+1$.
By the assumption of $H\in \mathbb{H}(\chi_1(\mathbf{x}))$ such that $|H|\ge b-1$ and $\{1,\,n\}\bigcap H\neq \emptyset$, we have $$w_b(\mathbf{c})=w_b(\mathbf{x}).$$
Hence $$ d_{b}(C)\leq w_b(\mathbf{c})=w_b(\mathbf{x})= d_b(C_1)$$
and $$d_b(C_1)= d_{b}(C),$$
which means $C$ is a $b$-symbol AMDS code.
\end{proof}

\begin{Remark}\label{5.6}
The $[n, n-1]_q$-linear codes $C_1$ and $C_2$ satisfying the assumption in Theorem~\ref{1.9a} can be found easily.
For example, we can assume that $C_1=C_2$ are both the dual code of the $[n, 1]_q$-linear code generated by $$(1,1,\cdots,1).$$
Let $$\mathbf{x}=(1,p-1,0,\cdots,0)$$ where $p=\mathrm{char}(\mathbb{F}_q)$. Since the length of $C_1$ can be any large integer, we know that the length $2n$ of the $b$-symbol AMDS code constructed by Theorem~\ref{1.9a} also can be any large integer.
We give the following example to explain this construction.
\end{Remark}

\begin{Example}
Assume $n=4$, $b=3$, $q=3$ and $\mathbb{F}_3=\{0,\,1,\,-1\}$.
Let
$C_{1}=C_{2}$ be the $[4,3]_3$-linear code with generator matrix
 $$G=\left(\begin{array}{cccc}
1  &-1  &0&0\\
1  & 0&-1&0\\
1  & 0&0&-1
\end{array}\right).$$
Then we have $$d_3(C_1)=d_3(C_2)=4.$$
Hence $C_1$ and $C_2$ are both $3$-symbol MDS codes and satisfy the condition in Theorem~\ref{1.9a}, where we take $$\mathbf{x}=(1,-1,0,0).$$
Let $C=[C_{1}, C_{2}]\cdot A$ where $A=\left(\begin{array}{cc}
1  &1 \\
1  & -1
\end{array}\right).$
By Corollary~\ref{5.3ab}, we have $4\leq d_3(C)\leq 5$.
Since $$[2\mathbf{x},\mathbf{0}]=(-1,1,0,0,0,0,0,0)\in C,$$
we have $d_3(C)=4$ and $C$ is a $3$-symbol AMDS code.
\end{Example}

\section{The $b$-symbol distances of two class of linear codes}
 The results of Section 5 involve $[n, n-1]_q$-linear codes and $[n, n-2]_q$-linear codes which are $b$-symbol MDS codes.
We calculate all the minimum $b$-symbol distances of $[n, n-1]_q$-linear codes and $[n, n-2]_q$-linear codes for $1\leq b\leq n-1$ in this section.
\begin{Theorem}\label{1.10}
Let $C$ be an $[n,n-1]_q$-linear code for $n\ge 2$. Then one of the following statements holds:
\begin{description}
  \item[(a)] $d_b(C)= b+1$  for every $1\leq b\leq n-1$, in this case $C$ is a $b$-symbol MDS code.
  \item[(b)] $d_b(C)= b$  for every $1\leq b\leq n-1$, in this case $C$ is  a $b$-symbol AMDS code.
\end{description}
\end{Theorem}

\begin{proof}
By the Singleton-like bound for the minimum Hamming distance, we have $1\leq d_1(C)\leq n-(n-1)+1=2$.
Hence $d_1(C)=1$ or $d_1(C)=2$.

Suppose $d_1(C)=2$, then we know that $$d_b(C)\ge b+1$$ for every $1\leq b\leq n-1$. Hence $C$ is a $b$-symbol MDS code for every $1\leq b\leq n-1$ by the Singleton-like bound for the minimum $b$-symbol distance.

Suppose $d_1(C)=1$, then we know that $$d_b(C)=b$$ for every $1\leq b\leq n-1$. Hence $C$ is a $b$-symbol AMDS code for every $1\leq b\leq n-1$ by the Singleton-like bound for the minimum $b$-symbol distance.
\end{proof}

\begin{Remark}
It is easy to see that there exists an $[n,n-1]_q$-linear code satisfying each of the conditions in Theorem~\ref{1.10}.
\end{Remark}

\begin{Theorem}\label{1.11}
Let $C$ be an $[n,n-2]_q$-linear code for $n\ge 3$, then one of the following statements holds:
\begin{description}
  \item[(a)] $d_b(C)=b$ for every $1\leq b\leq n-1$.
  \item[(b)] $d_b(C)=b+1$ for every $1\leq b\leq n-1$, in this case $C$ is a $b$-symbol AMDS code.
  \item[(c)] $d_1(C)=2$ and $d_b(C)=\min\{b+2, n\}$ for every $2\leq b\leq n-1$.
  \item[(d)] $d_b(C)=\min\{b+2, n\}$ for every $1\leq b\leq n-1$, in this case $C$ is a $b$-symbol MDS code.
  \end{description}
\end{Theorem}

\begin{proof}
By the Singleton-like bound for the minimum Hamming distance, we have $$1\leq d_1(C)\leq n-(n-2)+1=3.$$
Hence $d_1(C)=1$, $d_1(C)=2$ or $d_1(C)=3$.

Suppose $d_1(C)=1$, then we know that $$d_b(C)=b$$ for $1\leq b\leq n-1$.
This gives the case {\bf (a)}.

Suppose $d_1(C)=2$.  If there exists $\mathbf{c}\in C$ such that $w_1(\mathbf{c})=2$ and $w_2(\mathbf{c})=3$, then $d_b(C)=b+1$ for $1\leq b\leq n-1$. This yields the case {\bf (b)}.

If $w_2(\mathbf{c})=4$ for any $\mathbf{c}\in C$ such that $w_1(\mathbf{c})=d_1(C)=2$, then $$d_2(C)=4$$ and $C$ is a $2$-symbol MDS code.
By Lemma~\ref{bmdss}, we have that $C$ is a $b$-symbol MDS code for $2\leq b\leq n-1$.
Hence $$d_b(C)=\min\{b+2, n\}$$ for $2\leq b\leq n-1$. This is the case {\bf (c)}.

Suppose $d_1(C)=3$, then we know that $$d_b(C)\ge \min\{b+2, n\}$$ for every $1\leq b\leq n-1$ by Lemma~\ref{relationship-1}. Hence $$d_b(C)=\min\{b+2, n\}$$ for $1\leq b\leq n-1$ and $C$ is a $b$-symbol MDS code for every $1\leq b\leq n-1$. This gives the case {\bf (d)}.
\end{proof}

We give the following examples to show that there exists an $[n,n-2]_q$-linear code satisfying each of the conditions of Theorem~\ref{1.11}.

\begin{Example}
 For $n\ge 3$, let $C_1$ be the $[n,n-2]_q$-linear code with parity check matrix $$H_1=\left(\begin{array}{ccccc}
1  &0  &0&\cdots&0\\
0  & 1&0&\cdots&0
\end{array}\right),$$
then $C_1$ satisfies the condition $\mathbf{ (a)}$ of Theorem~\ref{1.11}.
\end{Example}

\begin{Example}
For $n\ge 4$, let $C_2$ be the $[n,n-2]_q$-linear code with parity check matrix $$H_2=\left(\begin{array}{ccccc}
1  &1  &0&\cdots&0\\
0  & 0&1&\cdots&1
\end{array}\right),$$
then we know $d_1(C)= 2$.
Assume $p=\mathrm{char}(\mathbb{F}_q)$, then $$(1,p-1,0,\cdots,0)\in C_2.$$
Hence $d_b(C)=b+1$ for every $1\leq b\leq n-1$ and $C_2$ satisfies the condition $\mathbf{ (b)}$ of Theorem~\ref{1.11}.
\end{Example}

\begin{Example}
 For $n\ge 4$ and $2|n$,
let $C_3$ be the $[n,n-2]_q$-linear code with parity check matrix $$H_3=\left(\begin{array}{ccccccc}
1  &0  &1&0&\cdots&1&0\\
0  & 1&0&1&\cdots&0&1
\end{array}\right),$$
then we know $d_1(C_3)=2$.
By the choice of $H_3$, we know that $w_2(\mathbf{c})=4$ for any $\mathbf{c}\in C$ such that $w_1(\mathbf{c})=2$.
Assume $p=\mathrm{char}(\mathbb{F}_q)$.
Since $(1,0,p-1,0,\cdots,0)\in C_3,$ we have $$d_2(C_3)=4.$$
Hence $d_b(C_3)=\min\{b+2, n\}$ for every $2\leq b\leq n-1$ and $C_3$ satisfies the condition $\mathbf{ (c)}$ of Theorem~\ref{1.11}.
\end{Example}

\begin{Example}
 For $n\ge 4$ and $2\nmid n$,
let $C_4$ be the $[n,n-2]_q$-linear code with parity check matrix $$H_4=\left(\begin{array}{cccccccc}
1  &0  &1&0&\cdots&1&0&1\\
0  & 1&0&1&\cdots&0&1&1
\end{array}\right),$$
then we know $d_1(C_4)=2$.
By the choice of $H_4$, we know that $w_2(\mathbf{c})=4$ for any $\mathbf{c}\in C$ such that $w_1(\mathbf{c})=2$.
Assume $p=\mathrm{char}(\mathbb{F}_q)$. Since $(1,0,p-1,0,\cdots,0)\in C_4,$ we have $$d_2(C_4)=4.$$
Hence $d_b(C_4)=\min\{b+2, n\}$ for every $2\leq b\leq n-1$ and $C_4$ satisfies the condition $\mathbf{ (c)}$ of Theorem~\ref{1.11}.
\end{Example}

\begin{Example}
For $3\leq n\leq q$, assume $\alpha_1,\,\alpha_2,\,\cdots,\,\alpha_n$ are $n$ distinct elements of $\mathbb{F}_q$.
Let $C_5$ be the $[n,n-2]_q$-linear code with parity check matrix $$H_5=\left(\begin{array}{cccc}
1  &1  &\cdots&1\\
\alpha_1  & \alpha_2&\cdots&\alpha_n
\end{array}\right),$$
then we know $C_5$ is a $1$-symbol MDS code.
Hence $d_b(C_5)=\min\{b+2,n\}$ for every $1\leq b\leq n-1$ and $C_5$ satisfies the condition $\mathbf{ (d)}$ of Theorem~\ref{1.11}.
\end{Example}

\vskip 4mm

\noindent {\bf Acknowledgement.} This work was supported by NSFC (Grant Nos. 12271199, 12171191), The Fundamental Research Funds for the Central Universities (Grant No. 30106220482) and Nanyang Technological University Research (Grant No. 04INS000047C230GRT01).



\begin{thebibliography}{99}
\bibitem{BC1} P. Beelen, ``A note on the generalized Hamming weights of Reed-Muller
codes," Applicable Algebra in Engineering Communication Computing, vol. 30, no. 3, pp. 233-242, 2019.

\bibitem{BC} T. Berger and P. Charpin, ``The automorphism group of generalized Reed-Muller codes," Discrete Mathematics, vol. 117, pp. 1-17, 1993.

\bibitem{BN} T. Blackmore and G. H. Norton, ``Matrix-product codes over $\mathbb{F}_q$," Applicable Algebra in Engineering Communication Computing, vol. 12, no. 6, pp. 477-500, 2001.


\bibitem{CB} Y. Cassuto and M. Blaum, ``Codes for symbol-pair read channels," IEEE Transactions on Information Theory, vol. 57, no. 12, pp. 8011-8020, 2011.

\bibitem{CJK} Y. M. Chee, L. Ji, H. M. Kiah, C. Wang and J. Yin, ``Maximum distance separable codes for symbol-pair read channels," IEEE Transactions on Information Theory, vol. 59, no. 11, pp. 7259-7267, 2013.


\bibitem{CLL} B. Chen, L. Lin and H. Liu, ``Constacyclic symbol-pair codes: lower bounds and optimal constructions," IEEE Transactions on Information Theory, vol. 63, no. 12, pp. 7661-7666, 2017.

\bibitem{CL} B. Chen and H. Liu, ``New bounds on the code size of symbol-pair codes," IEEE Transactions on Information Theory, vol. 69, no. 2, pp. 941-950, 2023.

\bibitem{DG} P. Delsarte, J. M. Goethals and F.J. Mac Williams, ``On generalized Reed-Muller codes and their relatives," Information and Control, Vol. 16, no. 5, pp. 403-442, 1970.

\bibitem{DZ} B. Ding, T. Zhang and G. Ge, ``Maximum distance separable codes for $b$-symbol read channels," Finite Fields and Their Applications, vol. 49, pp. 180-197, 2018.

\bibitem{D}B. Ding, G. Ge, J. Zhang, T. Zhang and Y. Zhang, ``New constructions of MDS symbol-pair codes," Designs, Codes and Cryptography, vol. 86, no. 4, pp. 841-859, 2018.

\bibitem{DNSS} H. Q. Dinh, B. T. Nguyen, A. K. Singh and S. Sriboonchitta, ``On the symbol-pair distance of repeated-root constacyclic codes of prime power lengths," IEEE Transactions on Information Theory, vol. 64, no. 4, pp. 2417-2430, 2017.

\bibitem{DWLS}	 H. Q. Dinh, X. Wang, H. Liu and S. Sriboonchitta, ``On the symbol-pair distances of repeated-root constacyclic codes of length $2p^s$," Discrete Mathematics, vol. 342, no. 11, pp. 3062-3078, 2019.


\bibitem{FLL} Y. Fan, S Ling. and H. Liu, ``Matrix product codes over finite commutative Frobenius rings," Designs, Codes and Cryptography, vol. 71, pp. 201-227, 2014.

\bibitem{FLL1} Y. Fan, S Ling. and H. Liu, ``Homogeneous weights of matrix product codes over finite principal ideal rings," Finite Fields and Their Applications, vol. 64, pp. 247-267, 2014.

\bibitem{HP} P. Heijnen and R. Pellikaan, ``Generalized Hamming weights of $q$-ary Reed-Muller codes," IEEE Transactions on Information Theory, vol. 44, no. 1, pp. 181-196, 1998.



\bibitem{HL}F. Hernando, K. Lally and D. Ruano, ``Construction and decoding of matrix-product codes from nested codes," Applicable Algebra in Engineering, Communication and Computing, vol. 20, pp. 497-507, 2009.





\bibitem{HR} F. Hernando and D. Ruano, ``Decoding of matrix-product codes,"  Journal of Algebra and Its Applications, vol. 12, no. 4, pp. 1-15, 2013.





\bibitem{H} G. Hughes, ``Constacyclic codes, cocycles and a $(u + v|u-v)$ construction," IEEE Transactions on Information Theory, vol. 46, no. 2, pp. 674-680, 2000.

\bibitem{KZL} X. Kai, S. Zhu and P. Li, ``A construction of new MDS symbol-pair codes," IEEE Transactions on Information Theory, vol. 61, no. 11, pp. 5828-5834, 2015.

\bibitem{KZZLC} X. Kai, S. Zhu, Y. Zhao, H. Luo and Z. Chen, ``New MDS symbol-pair codes from repeated root codes," IEEE Communications Letters, vol. 22, no. 3, pp. 462-465, 2018.

\bibitem{KLP}  T. Kasamt, S. Lin, and W. W. Peterson, ``Polynomial codes," IEEE Transactions on Information Theory, vol. 14, pp. 807-814, 1968.

\bibitem{KLP1}  T. Kasamt, S. Lin, and W. W. Peterson, ``New generalisations of the Reed-Muller codes Part I Primitive codes," IEEE Transactions on Information Theory, vol. 14, no, 2, pp. 189-199, 1968.

\bibitem{KLP2}  T. Kasamt, S. Lin, and W. W. Peterson, ``New generalisations of the Reed-Muller codes Part II: Nonprimitive Codes," IEEE Transactions on Information Theory, vol. 14, no, 2, pp. 199-205, 1968.


\bibitem{CCDFM} Y. Cao, Y. Cao, H. Q. Dinh, F. W. Fu and P. Maneejuk, ``On matrix-product structure of repeated-root constacyclic codes over finite fields," Discrete Mathematics, vol. 343, no. 4, pp. 111768, 2020.

\bibitem{LG} S. Li and G. Ge, ``Constructions of maximum distance separable symbol-pair codes using cyclic and constacyclic codes," Designs, Codes and Cryptography, vol. 84, no. 3, pp. 359-372, 2017.

\bibitem{LL} H. Liu and J. Liu, ``Homogeneous metric and matrix product codes over finite commutative principal ideal rings," Finite Fields and Their Applications, vol. 64,  pp. 101666, 2020.

\bibitem{LP1} H. Liu and X. Pan, ``Galois hulls of linear codes over finite fields," Designs, Codes and Crytography, vol. 88, no. 2, pp. 241-255, 2020.

\bibitem{LP}  H. Liu and X. Pan, ``Generalized pair weights of linear codes and linear isomorphisms preserving pair weights," IEEE Transactions on Information Theory, vol. 68, no. 1, pp. 105-117, 2022.

\bibitem{LP2} H. Liu and X. Pan, ``Generalized $b$-symbol weights of linear codes and $b$-symbol MDS Codes," IEEE Transactions on Information Theory, vol. 69, no. 4, pp. 2311-2323, 2023.

\bibitem{LEL} G. Luo, M. F. Ezerman, S. Ling and X. Pan, ``New families of MDS symbol-pair codes from matrix-product codes," IEEE Transactions on Information Theory, vol. 69, no. 3, pp. 1567-1587, 2023

\bibitem{ML1} J. Ma and J. Luo, ``MDS symbol-pair codes from repeated-root cyclic codes," Designs, Codes and Cryptography, vol. 90, no. 1, pp. 121-137, 2022.

\bibitem{ML2}  J. Ma and J. Luo, ``“Constructions of MDS symbol-pair codes with minimum distance seven or eight," Designs, Codes and Cryptography, vol. 90, no. 10, pp. 2337-2359, 2022.


\bibitem{MS}  F. J. MacWilliams and N. J. A. Sloane, ``The Theory of Error-Correcting Codes," North-Holland Publishing Company, 1977.

\bibitem{OS} F. \"{O}zbudak and H. Stichtenoth, ``Note on Niederreiter-Xing's propagation rule for linear codes," Applicable Algebra in Engineering, Communication and Computing, vol. 13, pp. 53-56, 2002.

\bibitem{R} A. Romanov,  ``On perfect and Reed-Muller codes over finite fields,"  Problems of Information Transmission, vol. 57, pp. 199-211, 2021.

\bibitem{SO} M. Shi, F. \"{O}zbudak and P. Sol\'{e}, ``Geometric approach to $b$-symbol Hamming weights of cyclic codes," IEEE Transactions on Information Theory, vol. 67, no. 6, pp. 3735-3751, 2021.



\bibitem{VV} E. van Eupen and J. H. van Lint, ``On the minimum distance of ternary cyclic codes," IEEE Transactions on Information Theory," vol. 39, no. 2, pp. 409-422, 1993.


\bibitem{V} B. Van Asch, ``Matrix-product codes over finite chain rings," Applicable Algebra in Engineering, Communication and Computing, vol. 19, no. 1, pp. 39-49, 2008.


\bibitem{V} G. Vega, ``The $b$-symbol weight distributions of all semiprimitive irreducible cyclic codes," Designs, Codes and Cryptography, 2023 (doi.org/10.1007/s10623-023-01193-w).



\bibitem{YBS}	E. Yaakobi, J. Bruck, and P. H. Siegel, ``Constructions and decoding of cyclic codes over $b$-symbol read channels," IEEE Transactions on Information Theory, vol. 62, no. 4, pp. 1541-1551, 2016.

\bibitem{ZSO} H. Zhu, M. Shi and F. \"{O}zbudak, ``Complete $b$-symbol weight distribution of some irreducible cyclic codes," Designs, Codes and Cryptography, vol. 90, pp. 1113-1125, 2022.





\end{thebibliography}
\end{document}